\documentclass[letterpaper, 10 pt, conference]{ieeeconf}  

\IEEEoverridecommandlockouts                              

\usepackage{amssymb,amsfonts,amsmath}
\usepackage{graphicx}
\usepackage{paralist}
\usepackage{dsfont}
\usepackage{booktabs}
\usepackage{subfigure}
\usepackage{eso-pic}
\usepackage{epsfig}
\usepackage {url}
\usepackage{epstopdf}
\usepackage{tikz}
\DeclareSymbolFont{symbolsC}{U}{pxsyc}{m}{n}
\SetSymbolFont{symbolsC}{bold}{U}{pxsyc}{bx}{n}
\DeclareFontSubstitution{U}{pxsyc}{m}{n}
\DeclareMathSymbol{\medcirc}{\mathbin}{symbolsC}{7}
\usetikzlibrary{arrows,calc,decorations.markings,math,arrows.meta}
\usepackage{algorithm}
\usepackage{algorithmicx}
\usepackage[noend]{algpseudocode}
\algnewcommand\algorithmicinput{\textbf{Input:}}
\algnewcommand\Input{\item[\algorithmicinput]}
\algnewcommand\algorithmicoutput{\textbf{Output:}}
\algnewcommand\Output{\item[\algorithmicoutput]}
\algnewcommand{\Initialize}[1]{%
	\State \textbf{Initialize:}
	\Statex \hspace*{\algorithmicindent}\parbox[t]{.8\linewidth}{\raggedright #1}
}
\newtheorem{theorem}{Theorem}[section]
\newtheorem{lemma}[theorem]{Lemma}

\newtheorem{problem}[theorem]{Problem}

\newtheorem{corollary}[theorem]{Corollary}

\newtheorem{definition}[theorem]{Definition}
\newtheorem{example}[theorem]{Example}

\newtheorem{remark}[theorem]{Remark}
\newtheorem{assumption}[theorem]{Assumption}
\numberwithin{equation}{section}

\newcommand{\ie}{{\it i.e.}}

\DeclareMathOperator{\diff}{d}

\title{\LARGE \bf
	Verification of Switched Stochastic Systems via Barrier Certificates*
}

\author{Mahathi Anand$^{\dagger}$, Pushpak Jagtap$^{\dagger}$ and Majid Zamani
	\thanks{*This work was supported in part by the H2020 ERC
		Starting Grant AutoCPS (grant agreement No. 804639), the German Research Foundation (DFG) through the grant ZA 873/1-1 and Research Training Group 2428, and the TUM International Graduate School of Science and Engineering (IGSSE).}
	\thanks{$^{\dagger}$Authors contributed equally to this work.}
	\thanks{M. Anand is with the Computer Science Department, Ludwig Maximilian University of Munich, Germany. P. Jagtap is with the Department of Electrical and Computer Engineering, Technical University of Munich, Germany. M. Zamani is with the Computer Science Department, University of Colorado Boulder, USA. M. Zamani is with the Computer Science Department, Ludwig Maximilian University of Munich, Germany. Emails:  {\tt\small niloofar.jahanshahi@lmu.de}, {\tt\small pushpak.jagtap@tum.de}, {\tt\small majid.zamani@colorado.edu}.}%
}

\begin{document}

	\maketitle
	\thispagestyle{empty}
	\pagestyle{empty}

	\begin{abstract}
		
		The paper presents a methodology for temporal logic verification of continuous-time switched stochastic systems. Our goal is to find the lower bound on the probability that a complex temporal property is satisfied over a finite time horizon. The required temporal properties of the system are expressed using a fragment of linear temporal logic, called safe-LTL with respect to finite traces. Our approach combines automata-based verification and the use of barrier certificates. It relies on decomposing the automaton associated with the negation of specification into a sequence of simpler reachability tasks and compute upper bounds for these reachability probabilities by means of common or multiple barrier certificates. Theoretical results are illustrated by applying a counter-example guided inductive synthesis framework to find barrier certificates.
		
	\end{abstract}

	\section{Introduction}
	Formal verification of dynamical systems against complex specifications has attracted significant attention in the past few years \cite{Tabuada.2009}. The verification problem becomes very challenging for the continuous-time continuous-space dynamical systems with noise and discrete dynamics. There are few results available on verification of continuous-time stochastic hybrid systems utilizing discrete approximations. Examples include probabilistic verification based on a discrete approximation for safety and reachability \cite{Koutsoukos.2008}, verification of stochastic hybrid systems described as piece-wise deterministic Markov processes \cite{Bujorianu.2003}, and safety verification of stochastic systems with state-dependent switching \cite{Prandini.2007}. However, these abstraction techniques are based on state set discretization and face the issue of discrete state set explosion.

	On the other hand, a discretization-free approach, based on barrier certificates, has been used for verifying stochastic hybrid systems against invariance property. Authors in \cite{4287147} used barrier certificate for safety verification of stochastic systems with probabilistic switching. Similar results are reported in \cite{Wisniewski.2018} for switched diffusion processes and piece-wise deterministic Markov processes. These results provide infinite time horizon guarantees. However, they require that barrier certificates exhibit a supermartingale property which presupposes stochastic stability and vanishing noise at the equilibrium point. 
	
	Our previous work \cite{Jagtap.2018} presents the idea of combining automata representation of a specification and barrier certificates, for the formal verification of discrete-time stochastic systems without requiring any stability assumption on the dynamics of the system. There, we only require \textit{$c$-martingale} property which can be fulfilled by unstable stochastic systems as well. The current manuscript follows the same direction to solve the problem of formal verification of continuous-time switched stochastic systems. 
	
	To the best of our knowledge, this paper is the first to use barrier certificates for the verification of continuous-time switched stochastic systems against a wide class of temporal logic properties. Our main contribution is to provide a systematic approach for computing lower bounds on the probability that a given switched stochastic system satisfies a fragment of linear temporal logic specifications, called safe-LTL, over finite time horizon. This is achieved by first decomposing the given specification into a sequence of simpler verification tasks based on the structure of the automaton corresponding to the negation of the specification. Then we use barrier certificates for computing probability bounds for these simple verification tasks which are then combined to get a (potentially conservative) lower bound on the probability of satisfying the original specification. We provide those probability bounds using common barrier certificates for arbitrary switching and using multiple barrier certificates for some probabilistic switching. The theoretical results are illustrated with the help of a numerical example.
	
	\section{Preliminaries}
	\subsection{Notations}
	We denote the set of real, positive real, nonnegative real, and positive integer numbers by $\mathbb{R}$, $\mathbb{R}^+$, $\mathbb{R}_0^+ $, and $\mathbb{N}$, respectively. We use $\mathbb{R}^n$ to denote an $n$-dimensional Euclidean space and $\mathbb{R}^{n \times m}$ to denote a space of real matrices with $n$ rows and $m$ columns. 
	Given a matrix $A\in\mathbb{R}^{n\times n}$, $\text{Tr}(A)$ represents trace of $A$ which is the sum of all diagonal elements of $A$.  
	Int($X$) represents interior of set $X$.
	\subsection{Switched Stochastic Systems}
	Let the triplet $(\Omega, \mathcal{F}, \mathbb{P})$ denote a probability space with a sample space $ \Omega $, filtration $ \mathcal{F} $, and the probability measure $ \mathbb{P} $. The filtration $\mathbb{F}= (\mathcal{F}_s)_{s\geq 0}$ satisfies the usual conditions of right continuity and completeness \cite{oksendal}. Let $ (W_s)_{s\geq0} $ be an $ r $-dimensional $ \mathbb{F} $-Brownian motion. 
	\begin{definition}
		A switched stochastic system is a tuple $S=(\mathbb{R}^n,M,\mathcal{M},F,G)$, where
		\begin{itemize}
			\item $\mathbb{R}^n$ is the state space;
			\item $M=\{1,2,\ldots,l\}$ is a finite set of modes;
			\item $\mathcal{M}$ is a subset of the set of all piece-wise constant c\`adl\`ag (i.e. right continuous and with left limits) functions of time from $\mathbb{R}^+_0$ to $M$, and characterized by a finite number of discontinuities on all bounded interval in $\mathbb{R}^+_0$;
			\item $F=\{f_1,f_2,\ldots,f_l\}$ and $G=\{g_1,g_2,\ldots,g_l\}$ are such that for any $m \in M$, $f_m:\mathbb{R}^n \rightarrow \mathbb{R}^n$  and $g_m:\mathbb{R}^n \rightarrow \mathbb{R}^{n \times r}$ satisfy standard local Lipschitz continuity and linear growth.		
		\end{itemize}
	\end{definition}
	A continuous-time stochastic process $\xi:\Omega\times \mathbb{R}_0^+\rightarrow\mathbb{R}^n$ is a solution process of $S$ if there exists $\mu \in \mathcal{M}$ satisfying
	\begin{equation}\label{eq:0}
	\diff\xi=f_\mu(\xi)\diff t+g_\mu(\xi)\diff W_t
	\end{equation}
	$\mathbb{P}$-almost surely ($\mathbb{P}$-a.s.) at each time $t \in \mathbb{R}_0^+$. For any given $m\in M$, we denote $S_m$ as the subsystem of $S$ defined as
	\begin{equation} \label{eq:1}
	\diff\xi=f_m(\xi)\diff t+g_m(\xi)\diff W_t.
	\end{equation}
The solution process of $S_m$ exists and is unique due to the assumptions on $f_m$ and $g_m$ \cite{oksendal}. We write $\xi^\mu(t)$ to denote the value of the solution process at time $t\in\mathbb{R}_0^+$ under the switching signal $\mu$, starting from the initial state $\xi^\mu(0)=x_0$ $\mathbb{P}$-a.s. Note that a solution process of $S_m$ is also a solution process of $S$ corresponding to constant switching signal $\mu(t)=m$, for all $t\in\mathbb{R}_0^+$. We also use $\xi^m(t)$ to denote the value of solution process of $S_m$ at time $t\in\mathbb{R}_0^+$, starting from the initial state of $\xi^m(0)=x_0$ $\mathbb{P}$-a.s.
The generator $\mathcal D$ of the solution process $\xi$ acting on function $B:\mathbb{R}^n \rightarrow \mathbb{R}$ is defined as follows.
	\begin{definition}\label{Generator} For any given $m\in M$, the generator $\mathcal{D}$ of the process $\xi$ of the stochastic system $S_m$ acting on function $B:\mathbb{R}^n \rightarrow \mathbb{R}$ is given by
		\begin{align}
		\mathcal{D}B(x_0,m)=\lim_{t\to 0} \dfrac{\mathbb{E}[B(\xi^m(t))| \xi^m(0)=x_0]-B(x_0)}{t}.   
		\end{align}
		By using Dynkin's formula \cite{lcg},  one has,
		\begin{align}
		\mathbb{E}[B&(\xi^m(t_2)|\xi^m(t_1)]\nonumber\\&=B(\xi^m(t_1)+\mathbb{E}[\int\limits_{t_1}^{t_2}\mathcal{D}B(\xi^m(t),m)\diff t|\xi^m(t_1)],
		\end{align}
		for $t_2\geq t_1\geq0$.
	\end{definition}
	
	\subsection{Linear Temporal Logic Over Finite Traces}
	In this paper, we consider specifications represented using linear temporal logic over finite traces, referred to as LTL$_F$ \cite{de2013linear}. LTL$_F$ uses the same syntax of LTL over infinite traces given in  \cite{baier2008principles}. Note that, the semantics of LTL$_F$ are however limited to interpretation over finite traces. The LTL$_F$ formulas over a set $ \Pi $ of atomic propositions are obtained as
	\begin{equation}
	\varphi ::=  \top \mid p \mid \neg \varphi \mid \varphi_1 \wedge \varphi_2 \mid \varphi_1 \vee\varphi_2\mid\medcirc \varphi  \mid  \lozenge\varphi \mid \square\varphi \mid \varphi_1\mathcal{U}\varphi_2, \nonumber
	\end{equation}
	where $p \in \Pi$, $\top$ represents true, $\medcirc $ is the next operator, $\lozenge$ is eventually, $\square$ is always, and $\mathcal{U}$ is until. The semantics of LTL$_F$ is given in terms of \textit{finite traces}, \textit{i.e.}, finite words $\sigma$, denoting a finite non-empty sequence of consecutive steps over $\Pi$. Detailed definitions for the semantics of LTL$_F$ have been omitted due to lack of space and can be found in \cite{Jagtap.2018}. 
	
	In this paper, we consider only safety properties \cite{baier2008principles}. In addition, we exclude the next ($\medcirc$) operator which enables us to describe behaviour of continuous trajectories using such properties. Hence, we use a subset of LTL$_F$ called safe-LTL$_{F\backslash\medcirc}$ as introduced in \cite{6942758}.
	
	\begin{definition}\label{safe_LTL_f} An LTL$_{F}$ formula is called a safe-LTL$_{F\backslash\medcirc}$ formula if it can be represented in a positive normal form, i.e., negations can only occur adjacent to atomic propositions, using temporal operator always ($\square$).
	\end{definition}
	
	Now, we define deterministic finite automata which can be used to represent LTL${_F}$ formulas. 
	\begin{definition}\label{DFA}
		A deterministic finite automaton $($DFA$)$ is a tuple $\mathcal{A}=(Q,Q_0,\Sigma,$ $\delta,F)$, where $Q$ is a finite set of states, $Q_0\subseteq Q$ is a set of initial states, $\Sigma$ is a finite set $($a.k.a. alphabet$)$, $\delta: Q\times\Sigma\rightarrow Q$ is a transition function, and $F\subseteq Q$ is a set of accepting states.
	\end{definition}
	We use notation $q\overset{\sigma}{\longrightarrow} q'$ to denote transition relation $(q,\sigma,q')\in\delta$.
	A finite word $\sigma=(\sigma_0,\sigma_1,\ldots,\sigma_{n-1})\in \Sigma^n$ is accepted by a DFA $\mathcal{A}$ if there exists a finite state run $q=(q_0,q_1,\ldots,q_{n})\in Q^{n+1}$ such that $q_0\in Q_0$, $q_k \overset{\sigma_k}{\longrightarrow} q_{k+1}$ for all $0\leq k< n$ and $q_{n}\in F$. The accepted language of $\mathcal{A}$, denoted by $\mathcal{L}(\mathcal{A})$, is the set of all words accepted by $\mathcal{A}$.
	According to \cite{de2015synthesis}, every LTL$_F$ formula $\varphi$ can be translated to a DFA $\mathcal{A}_\varphi$ that accepts the same language as $\varphi$, \textit{i.e.}, $\mathcal{L}(\varphi)=\mathcal{L}(\mathcal{A}_\varphi)$. Such DFA can be constructed explicitly or symbolically using existing tools: SPOT \cite{duret2016spot}, MONA \cite{henriksen1995mona}.
	
	\begin{remark}
		For a given LTL$_F$ formula $\varphi$ over atomic propositions $\Pi$, the associated DFA $\mathcal A_\varphi$ is usually constructed over the alphabet $\Sigma = 2^\Pi$.
		Without loss of generality, we work with the set of atomic propositions directly as the alphabet rather than its power set.
	\end{remark}

	\subsection{Property Satisfaction by Switched Stochastic Systems}
	For a given switched stochastic system $S=(\mathbb{R}^n,M,\mathcal{M},F,G)$ with dynamics \eqref{eq:0}, the solution processes over finite time intervals are connected to LTL$_{F\backslash \medcirc}$ formulas with the help of a measurable labeling function $L: \mathbb{R}^n \rightarrow \Pi$, where $\Pi$ is the set of atomic propositions.  
	\begin{definition}
		For a switched stochastic system $S=(\mathbb{R}^n,M,\mathcal{M},F,G)$ and the labeling function $L:\mathbb{R}^n \rightarrow \Pi$, a finite sequence $\sigma_{\xi}=(\sigma_0,\sigma_1,\ldots,\sigma_{n-1}) \in \Pi^{n}$ is a finite trace of the solution process $\xi$ over a finite time horizon $[0,T)\subset\mathbb{R}_0^+$ if there exists an associated time sequence $t_0,t_1,\ldots,t_{n-1}$ such that $t_0=0$, $t_n=T$, and for all $j \in \{0,1,\ldots,n-1\}$, $t_j \in \mathbb{R}_{0}^+$ following conditions hold
		\begin{itemize}
			\item $t_j<t_{j+1}$;
			\item $\xi^{\mu}(t_j)\in L^{-1}({\sigma_j})$;
			\item If $\sigma_j \neq \sigma_{j+1}$, then for some $t_j' \in [t_j,t_{j+1}]$, $\xi^{\mu}(t)\in L^{-1}({\sigma_j})$ for all $t \in (t_j,t_j')$; $\xi^{\mu}(t)\in L^{-1}({\sigma_{j+1}})$ for all $t \in (t_j',t_{j+1})$; and either $\xi^{\mu}(t_j')\in L^{-1}(\sigma_j)$ or $\xi^{\mu}(t_j')\in L^{-1}(\sigma_{j+1})$.
		\end{itemize}
	\end{definition}
	
	Next we define the probability that the solution process $\xi$ of the switched stochastic system  $S$ starting from some initial state $\xi^{\mu}(0)=x_0\in \mathbb{R}^n$ satisfies safe-LTL$_{F\backslash \medcirc}$ formula $\varphi$ over a finite time horizon $[0,T)\subset\mathbb{R}_0^+$.
	\begin{definition}
		Consider a switched stochastic system $S$ and a safe-LTL$_{F\backslash\medcirc}$ formula $\varphi$ over $\Pi$. Then $\mathbb{P}_{x_0}\{\sigma_{\xi} \models \varphi\}$ is the probability that $\varphi$ is satisfied by the solution process $\xi$ of the system $S$ starting from the initial value of $x_0\in \mathbb{R}^n$ over a finite time horizon $[0,T)\subset\mathbb{R}_0^+$.
	\end{definition}
	\begin{remark}
		The set of atomic propositions $\Pi=\{p_0,p_1,\ldots,p_N\}$ and the labeling function $L: \mathbb{R}^n \rightarrow \Pi$ provide a measurable partition of the state space $\mathbb{R}^n = \cup_{i=1}^N X_i$ as  $X_i:=L^{-1}(p_i)$. Without loss of generality, we assume that $X_i\neq \emptyset$ for any $i$.
	\end{remark}
	\subsection{Problem Formulation}
	\begin{problem} 
		Given a switched stochastic system $S=(\mathbb{R}^n,M,\mathcal{M},F,G)$ with dynamics \eqref{eq:0}, a safe-LTL$_{F\backslash\medcirc}$ over a set $\Pi=\{p_0,p_1,\ldots,p_N\}$ of atomic propositions, and a labeling function $L:\mathbb{R}^n \rightarrow \Pi$, compute a lower bound on the probability $\mathbb{P}_{x_0}\{\sigma_{\xi} \models \varphi\}$ for all $x_0\in L^{-1}(p_i)$ for $i\in\{0,1,\ldots,N\}$.
	\end{problem}

	\begin{example} \label{example}
		Consider a switched stochastic system $S=(\mathbb{R}^2,M,\mathcal{M},F,G)$ with $M=\{1,2\}$, and dynamics
		\begin{align}
		&S_1:
		\begin{matrix}
		\diff\xi_1=-0.1\xi_2^2\diff t+\diff W_{1t},\\
		\ \ \diff\xi_2=-0.1\xi_1\xi_2\diff t+\diff W_{2t};
		\end{matrix}\\
		&S_2:
		\begin{matrix}
		\diff\xi_1=-0.1\xi_1^2\diff t+\diff W_{1t},\\
		\ \ \diff\xi_2=-0.1\xi_1\xi_2\diff t+\diff W_{2t}.
		\end{matrix}
		\end{align}
		Let the regions of interest be given as 
		\begin{align*}
		X_0&=\{(x_1,x_2)\in \mathbb{R}^2 \mid (x_1+5)^2+x_2^2\leq 2.5\},\\
		X_1&=\{(x_1,x_2)\in \mathbb{R}^2 \mid (x_1-5)^2+(x_2-5)^2\leq 3 \},\\
		X_2&=\{(x_1,x_2)\in \mathbb{R}^2 \mid (x_1-4)^2+(x_2+3)^2\leq 2 \}, \text{ and }\\
		X_3& = \mathbb{R}^2\setminus (X_0\cup X_1\cup X_2).
		\end{align*}
		The sets $X_0$, $X_1$, $X_2$, and $X_3$ are shown in Figure 1(a).
		
		The set of atomic propositions is given by $\Pi=\{p_0,p_1,p_2,p_3\}$, with labeling function $L(x) = p_i$ for any $x\in X_i$, $i\in\{0,1,2,3\}$. Given an initial state, we are interested in computing a tight lower bound on the probability that the solution process of $S$ over time horizon $[0,T)\subset\mathbb{R}_0^+$ satisfies the following specification: 
		\begin{itemize}
			\item If it starts in $X_0$, it will always stay away from $X_1$ or always stay away from $X_2$ within time horizon $[0,T)\subset\mathbb{R}_0^+$. If it starts in $X_2$, it will always stay away from $X_1$ within time horizon $[0,T)\subset\mathbb{R}_0^+$. 
		\end{itemize}
		This property can be expressed by the safe-LTL$_{F\backslash \medcirc}$ formula
		\begin{equation}
		\label{LTL_f}
		\varphi=(p_0\wedge(\square \neg p_1 \vee \square\neg p_2))\vee (p_2\wedge\square\neg p_1).
		\end{equation}
		The DFA corresponding to the negation of the safe-LTL$_F$ formula $\varphi$ in \eqref{LTL_f} is shown in Figure 1(b).
	\end{example}
		\begin{figure}[t!]
		\centering
		\subfigure[]{\includegraphics[scale=0.5, height = 5.2cm]{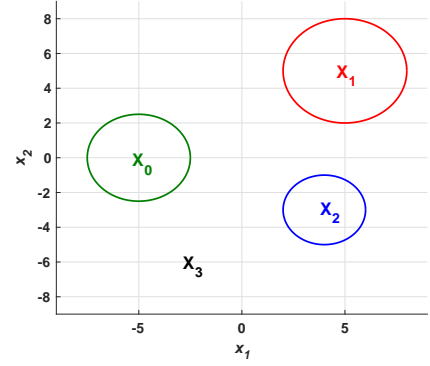}}
		\subfigure[]{\includegraphics[scale=0.5, height = 4.2cm]{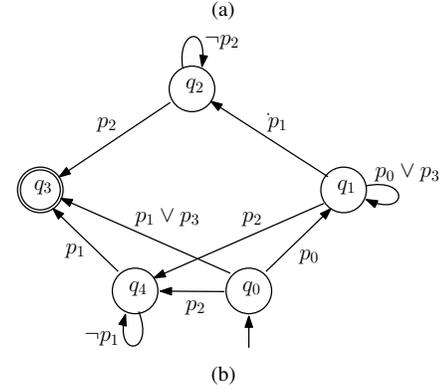}}
		\caption{(a) State space and regions of interest for Example~2.11, (b) DFA $\mathcal{A}_{\neg \varphi}$ that accepts all traces satisfying $\neg \varphi$ where $\varphi$ is given in \eqref{LTL_f}.}
	\end{figure}
	\section{Barrier Certificates}
	We recall that a function $B: \mathbb{R}^n \rightarrow \mathbb{R}$
	is a supermartingale for system $S$ if 
	$\mathbb{E}[B(\xi^{\mu}(t_2)\mid\xi^{\mu}(t_1)] \leq B(\xi^{\mu}(t_1))$ 
	for all $t_2 \geq t_1$. Although this condition is useful for the verification of stochastic systems \cite{4287147} for infinite horizons, it pre-supposes stochastic stability of the system and such a function may not exist in general. Hence, we use a relaxation of supermartingale condition called \textit{$c$-martingale} which enables us to provide results over a finite time horizon \cite{steinhardt2012finite} without any stability assumption.
	\begin{definition}
		A function $B: \mathbb{R}^n \rightarrow \mathbb{R}$ is a $c$-martingale for the system $S$ if
		$$\mathbb{E}[B(\xi^{\mu}(t_2)\mid\xi^{\mu}(t_1)] \leq B(\xi^{\mu}(t_1))+\int_{t_1}^{t_2}c(t)\diff t$$  
		for all $t_2 \geq t_1$, where $c$ is a function of time.
	\end{definition}
	
	We provide the following lemma and use it in the sequel. This lemma is a direct consequence of \cite[Theorem 1]{kushner1965stability}.
	\begin{lemma}\label{lemma_1}
		Let $B: \mathbb{R}^n\rightarrow \mathbb R_0^{+}$ be a non-negative $c$-martingale for the system $S$. Then for any constant $\lambda>0$ and any initial condition $x_0\in \mathbb{R}^n$, 
		\begin{equation}
		\mathbb{P}\{\sup_{0\leq t < T}B(\xi^{\mu}(t))\geq\lambda\mid x(0)=x_0\}\leq\frac{B(x_0)+\int_{0}^{T}c(t)\diff t}{\lambda}.
		\label{eq:lemma}
		\end{equation}
	\end{lemma}
	
The next two theorems provide inequalities on barrier certificates to give an upper bound on reachability probability. These theorems have been inspired by the results in \cite{4287147} that uses supermartingales for safety verification of continuous-time switching diffusion systems.
	
	\begin{theorem} \label{barrier1}
		Consider a switched stochastic system $S=(\mathbb{R}^n,M,\mathcal{M},F,G)$ with dynamics \eqref{eq:0} and sets $X_0, X_1 \subseteq \mathbb{R}^n$. Suppose there exists a twice differentiable function $B: \mathbb{R}^n\rightarrow \mathbb R_0^{+}$ and constants $c\geq 0$ and $\gamma \in [0,1]$, such that
		\begin{align} 
		&B(x) \leq \gamma \quad \forall x \in X_0,\label{eq:th1}\\
		&B(x) \geq 1 \quad \forall x \in X_1,\label{eq:th2}
		\end{align}
		\begin{align}
		&\frac{\partial B}{\partial x}(x)f_m(x)+\frac{1}{2}\text{Tr}\big(g_m^T(x)\frac{\partial^2 B}{\partial x^2}(x)g_m(x)\big) \leq c \nonumber\\
		&\quad\quad\quad\quad\quad\quad\quad\quad\quad\forall x \in \mathbb{R}^n, \forall m\in M.\label{eq:th3}
		\end{align}
		Then the probability that the solution process $\xi$ of the system $S$ starts from initial state $\xi^{\mu}(0)=x_0\in X_0$ and reaches $X_1$ within time horizon $[0,T)\subset\mathbb{R}_0^+$ is upper bounded by $\gamma+c T$.
	\end{theorem}
	\begin{proof}
		The generator associated with the system $S_m$ is given by 
		$$\mathcal{D}B(x,m)=\frac{\partial B}{\partial x}(x)f_m(x)+\frac{1}{2}\text{Tr}(g_m^T(x)\frac{\partial^2 B}{\partial x^2}(x)g_m(x)),$$
		where $m\in M$. Now, by using Dynkin's formula, we can show that $B(x)$ is a nonnegative $c$-martingale for all $m\in M$ and hence \eqref{eq:lemma} in Lemma \ref{lemma_1} holds. Using \eqref{eq:th1} and the fact that $X_1\subseteq\{x\in \mathbb{R}^n\mid B(x)\geq 1\}$, we have $\mathbb{P}\{\xi^{\mu}(t) \in X_1 \  \text{for some} \  0\leq t<T \mid \xi^{\mu}(0)=x_0\} \leq \mathbb{P}\{\sup_{0\leq t < T}B(\xi^{\mu}(t)) \geq 1 \mid \xi^{\mu}(0)=x_0\} \leq B(x_0)+cT \leq \gamma+cT$. 
		This concludes the proof.
	\end{proof} 
	If there exists a twice differentiable function $B:\mathbb{R}^n\rightarrow\mathbb{R}^+_0$ satisfying the conditions \eqref{eq:th1}-\eqref{eq:th3} of Theorem \ref{barrier1}, then we
	call it a common barrier certificate. In most of the cases, finding common barrier certificates may not be feasible or may result in conservative probability bounds. To alleviate these issues, we provide results using multiple barrier certificates for switched stochastic systems with a restricted set of switching signals as defined below. 
	
	Consider a switched stochastic system $S$ as defined in \eqref{eq:0} and $m,m' \in M=\{1,2,\ldots,k\}$. At any instant $t$, the transition probability between modes is given by
	\begin{equation}\notag\label{tran_prob}
	\mathbb{P}\{(m,m'),t+h\}= 
	\begin{cases}
	\lambda_{mm'}(\xi^{\mu}(t))h, & \text{if }m \neq m', \\
	1+\lambda_{mm}(\xi^{\mu}(t))h, & \text{if }m=m', \\
	\end{cases}
	\end{equation}
	where $h>0$, $\lambda_{mm'}: \mathbb{R}^n \rightarrow \mathbb{R}$ is a bounded and Lipschitz continuous function representing transition rates such that $\lambda_{mm'}(x) \geq 0$ for all $x\in\mathbb{R}^n$ if $ m \neq m'$ and $\sum_{m' \in M}\lambda_{mm'}(x)=0$ for all $m\in M$. It is assumed that the transition from one mode to another is independent of the Wiener process $W_t$. 
	
	The next theorem provides conditions to obtain an upper bound on the reachability probability for switched stochastic systems using multiple barrier certificates.
	
	\begin{theorem} \label{barrier2}
		Consider a switched stochastic system $S=(\mathbb{R}^n,M,\mathcal{M},F,G)$ with dynamics \eqref{eq:0}, sets $X_0, X_1 \subseteq \mathbb{R}^n$, and the transition rates between two switching modes $m,m'\in M$ as $\lambda_{mm'}:\mathbb{R}^n\rightarrow\mathbb{R}$. Suppose there exists a set of twice differentiable functions $B_m:\mathbb{R}^n\rightarrow\mathbb{R}_0^+$, and constants $c\geq 0$ and $\gamma \in [0,1]$, such that 
		\begin{align} 
		&B_m(x) \leq \gamma \quad \forall x \in X_0,\label{eq:th21}\\
		&B_m(x) \geq 1 \quad  \forall x \in X_1, \label{eq:th22}\\
		&\frac{\partial{B_m}}{\partial{x}}(x)f_m(x)+\frac{1}{2}\text{Tr}(g_m^T(x)\frac{\partial^2{B_m}}{\partial x^2}g_m(x))\nonumber  
		\\ &\quad \quad \quad + \sum_{m' \in M} \lambda_{mm'}(x)B_{m'}(x) \leq c \quad\forall x \in \mathbb{R}^n.
		\label{eq:th23}
		\end{align}
		for all $m\in M$. Then the probability that the solution process $\xi$ of the system $S$ starts from initial state $\xi^{\mu}(0)=x_0\in X_0$ and reaches $X_1$ within time horizon $[0,T)\subset\mathbb{R}_0^+$ is upper bounded by $\gamma+c T$.
	\end{theorem}
	
	\begin{proof} The proof is similar to that of Theorem \ref{barrier1}. 
	\end{proof}
	\section{Decomposition into Sequential Reachability}
	Consider a DFA $\mathcal{A}_{\neg\varphi}=(Q,Q_0,\Pi,\delta,F)$ that accepts all finite words over $\Pi$ that satisfy $\neg\varphi$. 
	The sequence $\textbf{q}=(q_0,q_1,\ldots,q_n)\in Q^{n+1}$, $n\in\mathbb{N}$ is called an accepting state run if $q_0\in Q_0$, $q_n\in F$, and there exists a finite word $\sigma = (\sigma_0,\sigma_1,\ldots,\sigma_{n-1})\in\Pi^n$ such that $q_k \overset{\sigma_k}{\longrightarrow} q_{k+1}$ for all $k\in\{0,1,\ldots, n-1\}$. We denote the set of such finite words by $\sigma(\textbf{q})\subseteq \Pi^n$. 
	We also indicate the length of $\textbf{q}\in Q^{n+1}$ by $|\textbf{q}|$, which is $n+1$. Let $\mathcal{R}$ be the set of all finite accepting state runs starting from $p\in\Pi$ excluding self-loops, where
	\begin{equation}\notag
	\begin{split}
	\mathcal{R}\hspace{-.2em} := \hspace{-.2em} \{\textbf{q}\hspace{-.2em}  =\hspace{-.2em}  (q_0,q_1,\ldots,q_n)\hspace{-.2em} \in \hspace{-.2em} Q^{n+1} \mid q_n\hspace{-.2em} \in \hspace{-.2em} F, q_k\hspace{-.2em} \neq\hspace{-.2em}  q_{k+1},\forall k\hspace{-.2em} <\hspace{-.2em} n\}.
	\end{split}
	\end{equation}
	Computation of $\mathcal{R}$ can be done algorithmically by viewing $\mathcal{A}_{\neg\varphi}$ as a directed graph $\mathcal{G}=(\mathcal{V},\mathcal{E})$ with vertices $\mathcal{V}=Q$ and edges $\mathcal{E}\subseteq\mathcal{V}\times\mathcal{V}$ such that $(q,q')\in\mathcal{E}$ if and only if $q'\neq q$ and there exist $p\in\Pi$ such that $q\overset{p}{\longrightarrow} q'$. From the construction of the graph, it is obvious that the finite path in the graph starting from vertices $q_0\in Q_0$ and ending at $q_F\in F$ is an accepting state run $\textbf{q}$ of $\mathcal{A}_{\neg\varphi}$ without any self-loop and therefore belongs to $\mathcal{R}$. One can easily compute $\mathcal{R}$ using depth first search algorithm \cite{russell2003artificial}.
	
	For each $p \in \Pi$, we define a set $\mathcal{R}^p$ as
	\begin{equation}
	\label{eq:runs}
	\mathcal{R}^p := \{\textbf{q} = (q_0,q_1,\ldots,q_n)\in\mathcal R \mid \sigma(q_0,q_1)=p\}.
	\end{equation}
	
	Decomposition into sequential reachability is performed as follows. For any $\textbf{q}=(q_0,q_1,\ldots,q_n) \in \mathcal{R}^p \ \forall p \in \Pi$, we define $\mathcal{P}^p(\textbf{q})$ as a set of all state runs of length 3,
	\begin{equation} \label{eq:reachability}
	\mathcal{P}^p(\mathbf{q}):=\{(q_k,q_{k+1},q_{k+2}) \mid 0 \leq k \leq n-2\}.
	\end{equation}

\begin{remark}
		Note that $\mathcal{P}^p(\textbf{q})=\emptyset$ for $|\textbf{q}|=2$. Any accepting state run of length $2$ begins from a subset of the state space that already satisfies $\neg\varphi$ and hence gives trivial zero probability for satisfying the specification, and is thus neglected in the sequel.
	\end{remark}
	\def\example{\par{\textit{Example 2.11}} \ignorespaces}
	\def\endexaple{}
	
	\begin{example}
		For safe-LTL$_{F\backslash\medcirc}$ formula $\varphi$ given in \eqref{LTL_f}, Figure 1(b) shows a DFA $\mathcal{A}_{\neg\varphi}$ that accepts all words that satisfy $\neg\varphi$. From Figure 1(b), we get $Q_0=\{q_0\}$ and $F=\{q_3\}$.
		The set of accepting state runs without self-loops is
		\begin{equation*}
		\mathcal{R}=\{(q_0,q_4,q_3),(q_0,q_1,q_2,q_3),(q_0,q_1,q_4,q_3),(q_0,q_3)\}.
		\end{equation*}
		The sets of $\mathcal{R}^p$ for $p \in \Pi$ are
		\begin{align*}
		&\mathcal{R}^{p_0}=\{(q_0,q_1,q_2,q_3),(q_0,q_1,q_4,q_3)\}, \ 
		\mathcal{R}^{p_1}=\{(q_0,q_3)\}, \\
		&\mathcal{R}^{p_2}=\{(q_0,q_4,q_3)\}, \ 
		\mathcal{R}^{p_3}=\{(q_0,q_3)\}. 
		\end{align*}
		The sets $\mathcal{P}^p(\textbf{q})$ for $\textbf{q}\in\mathcal{R}^p$ are as follows:
		\begin{align*}
		&\mathcal{P}^{p_0}(q_0,q_1,q_2,q_3)=\{(q_0,q_1,q_2),(q_1,q_2,q_3)\},\\
		&\mathcal{P}^{p_0}(q_0,q_1,q_4,q_3)=\{(q_0,q_1,q_4),(q_1,q_4,q_3)\}, \\
		&\mathcal{P}^{p_2}(q_0,q_4,q_3)\hspace{-0.2em}=\hspace{-0.2em}\{(q_0,q_4,q_3)\},\mathcal{P}^{p_1}(q_0,q_3)\hspace{-0.2em}=\hspace{-0.2em}\mathcal{P}^{p_3}(q_0,q_3)\hspace{-0.2em}=\hspace{-0.2em}\emptyset.
		\end{align*} 
	\end{example}
	\section{Computation of Probabilities Using Barrier Certificates}
	Having the set of state runs of lengths 3, we provide a systematic approach to compute lower bound on the probability that the solution process $\xi$ satisfies $\varphi$. Given the DFA $\mathcal{A}_{\neg\varphi}$ corresponding to specification $\neg\varphi$, we perform the computation of upper bound on reachability probability over each element of $\mathcal{P}^{p}(\textbf{q})$, $\textbf{q} \in \mathcal{R}^p$ using barrier certificates. 
	Next theorem provides an upper bound on the probability that the solution process $\xi$ satisfies the specification $\neg\varphi$.
	\begin{theorem} \label{upper_bound}
		For a given safe-LTL$_{F\backslash\medcirc}$ specification $\varphi$, let $\mathcal{A}_{\neg\varphi}$ be the DFA corresponding to its negation, $\mathcal R^p$ be the set defined in \eqref{eq:runs}, and $\mathcal{P}^p$ be the set of runs of length $3$ defined in \eqref{eq:reachability}.  Then the probability that the solution process of system $S$ starting from any initial state $x_0 \in L^{-1}(p)$ satisfies $\neg\varphi$ within time horizon $[0,T)$ is upper bounded by
		\begin{align} \label{eq:finalprob}
		&\mathbb{P}_{x_0}\hspace{-0.2em}\{\sigma_{\xi} \models\hspace{-0.2em} \neg\varphi\}\hspace{-0.2em} \le \hspace{-0.3em}\sum_{\bf q \in \mathcal{R}^p}\hspace{-0.2em}\prod\{(\gamma_\nu+c_\nu T)\hspace{-0.2em} \mid\hspace{-0.2em} \nu\hspace{-0.2em}=\hspace{-0.2em}(q,q',q'') \hspace{-0.2em}\in\hspace{-0.2em} \mathcal{P}^p(\bf{q})\},
		\end{align}
		where $\gamma_\nu+c_\nu T$ is the upper bound on the probability of the solution process of system $S$ starting from $X_0 := L^{-1}(\sigma(q,q'))$ and reaching $ X_1:=L^{-1}(\sigma(q',q''))$ within time horizon $[0,T)$ computed via Theorem~\ref{barrier1} (or Theorem~\ref{barrier2}). 
	\end{theorem}
	\begin{proof} For $p\in\Pi$, consider an accepting run $\textbf{q} \in \mathcal{R}^p$ and set $\mathcal{P}^p(\textbf{q})$ as defined in \eqref{eq:reachability}. For an element $\nu=(q,q',q'')\in \mathcal{P}^p(\textbf{q})$, the upper bound on the probability that solution processes of $S$ starting from $L^{-1}(\sigma(q,q'))$ and reaching $L^{-1}(\sigma(q',q''))$ within time horizon $T$ is given by $\gamma_\nu+c_\nu T$. This follows from Theorem \ref{barrier1} (or Theorem \ref{barrier2}). Now the upper bound on the probability that the trace of the solution process reaches the accepting state following the path corresponding to $\textbf{q}$ is given by the product of the probability bounds corresponding to all elements $\nu=(q,q',q'')\in \mathcal{P}^p(\textbf{q})$ and is given by
		\begin{equation}
		\label{eq:bound_1}
		\mathbb{P}\{ \sigma(\textbf{q})\models \neg \varphi \} \leq \prod \left \{(\gamma_\nu+c_\nu T)\mid\hspace{-0.2em} \nu \hspace{-0.2em}=\hspace{-0.2em} (q,q',q'')\hspace{-0.2em}\in\hspace{-0.2em} \mathcal{P}^p(\textbf{q})\right \}.
		\end{equation}
		The upper bound on the probability that the solution process of system $S$ starting from any initial state $x_0\in L^{-1}(p)$ violate $\varphi$ can be computed by summing the probability bounds for all possible accepting runs as computed in \eqref{eq:bound_1} and is given by
		\begin{align*} 
		&\mathbb{P}_{x_0}\hspace{-0.2em}\{\sigma_{\xi} \models\hspace{-0.2em} \neg\varphi\}\hspace{-0.2em} \le \hspace{-0.3em}\sum_{\bf q \in \mathcal{R}^p}\hspace{-0.2em}\prod\{(\gamma_\nu+c_\nu T)\hspace{-0.2em} \mid\hspace{-0.2em} \nu\hspace{-0.2em}=\hspace{-0.2em}(q,q',q'') \hspace{-0.2em}\in\hspace{-0.2em} \mathcal{P}^p(q)\}. \nonumber
		\end{align*}
	\end{proof}
	Theorem~\ref{upper_bound} enables us to decompose the specification 
	into a collection of sequential reachabilities, compute bounds on the reachability probabilities using Theorem~\ref{barrier1} (or Theorem \ref{barrier2}), and then combine the bounds in a sum-product expression.
	\begin{remark}
		In case we are unable to find barrier certificates for some of the elements $\nu \in \mathcal P^p(\textbf q)$ in \eqref{eq:finalprob}, we replace the related term $(\gamma_\nu+c_\nu T)$ by the pessimistic bound $1$. In order to get a non-trivial bound in \eqref{eq:finalprob}, at least one barrier certificate must be found for each $\textbf{q} \in\mathcal R^p$.
	\end{remark}
	
	\begin{corollary}
		\label{lower_bound}
		Given the result of Theorem \ref{upper_bound}, the probability that the solution process of $S$ starting from any $x_0 \in L^{-1}(p)$ over time horizon $[0,T)\subset\mathbb{R}_0^+$ satisfies safe-LTL$_{F\backslash\medcirc}$ specification $\varphi$ is lower-bounded by
		\begin{equation}
		\mathbb{P}_{x_0}\{\sigma_{\xi}\models \varphi \}\geq1-\mathbb{P}_{x_0}\{\sigma_{\xi}\models \neg \varphi \}.\nonumber
		\end{equation}
	\end{corollary}
	\section{Computation of Barrier Certificates}
	In this section, we provide the Counter-Example Guided Inductive Synthesis (CEGIS) framework for searching barrier certificates of specific forms satisfying conditions in Theorem~\ref{barrier1} (or Theorem \ref{barrier2}). The approach uses feasibility solvers for finding barrier certificates of a given parametric form using Satisfiability Modulo Theories (SMT) solvers such as Z3 \cite{z3} and dReal \cite{Gao.2013}. In order to use the CEGIS framework, we raise following assumption.
	\begin{assumption}\label{assum1}
		System $S$ has compact state-space $X\subset\mathbb{R}^n$ and partition sets $X_i \in L^{-1}(p_i)$, $i\in\{1,2,\ldots,N\}$ are bounded, semi-algebraic sets, \ie, they can be represented by polynomial equalities and inequalities.
	\end{assumption}
	\begin{remark}
		The assumption of compactness of state-space $X\subseteq\mathbb R^n$ can be supported by considering stopped process $\tilde{\xi}:\Omega\times \mathbb{R}_0^+\rightarrow X$ as
		\begin{equation*} 
		\tilde{\xi}^\mu(t):= 
		\begin{cases}  
		\xi^{\mu}(t),\quad \text{for} \ t<\tau,  \\
		\xi^{\mu}(\tau), \quad \text{for} \ t \geq \tau,
		\end{cases}
		\end{equation*}
		where $\tau$ is the first time of exit of the solution process $\xi$ of $S=(\mathbb{R}^n,M,\mathcal{M},F,G)$ from the open set Int($X$). Note that, in most cases, the generator corresponding to $\tilde\xi$ is identical to the one corresponding to $\xi$ over the set Int($X$), and is equal to zero outside of the set \cite{kushner}. Thus, the results in theorems \ref{barrier1} and \ref{barrier2} can be used for the systems with this assumption.
	\end{remark}
	
	The feasibility condition for the existence of common barrier certificate required in Theorem  \ref{barrier1} is provided in next lemma.
	
	\begin{lemma} \label{smt1}
		Consider a switched stochastic system $S=(X,M,\mathcal{M},F,G)$ with Assumption \ref{assum1}. Suppose sets $X_0$, $X_1$, and $X$ are bounded semi-algebraic sets. Suppose there exists a function $B(x)$, constants $\gamma \in [0, 1]$ and $c \geq 0$, such that the following expression is true
		\begin{align}\label{feas1}
		&\bigwedge_{x \in X} B(x) \geq 0 \bigwedge_{x \in X_0} B(x) \leq \gamma \bigwedge_{x \in X_1} B(x) \geq 1    \nonumber \\ 
		&\bigwedge_{m\in M}\Big(\bigwedge_{x \in X}\frac{\partial B}{\partial x}(x)f_m(x)  +\frac{1}{2}\text{Tr}\big(g_m^T(x)\frac{\partial^2 B}{\partial x^2}(x)g_m(x)\big) \leq c\Big).
		\end{align}
		Then $B(x)$ satisfies conditions in Theorem \ref{barrier1}.
	\end{lemma}
	One can easily obtain an analogous feasibility condition for the existence of multiple barrier certificates required in Theorem \ref{barrier2}.
	
	In order to utilize CEGIS framework, we consider a barrier certificate of the parametric form $B(a,x)= \sum_{i=1}^{k}a_i b_i(x)$ with some user-defined (nonlinear) basis functions $b_i(x)$ and unknown coefficients $a_i \in \mathbb{R}, i \in \{1, 2,\ldots, k\}$. With this choice of barrier certificate the feasibility expression \eqref{feas1} can be rewritten as
	\begin{align*}
	&\psi(a,x)\hspace{-0.2em}:=\hspace{-0.2em}\bigwedge_{x \in X} \hspace{-0.2em}B(a,x)\hspace{-0.2em} \geq \hspace{-0.2em}0 \bigwedge_{x \in X_0}\hspace{-0.2em} B(a,x)\hspace{-0.2em} \leq \hspace{-0.2em}\gamma \bigwedge_{x \in X_1}\hspace{-0.2em} B(a,x)\hspace{-0.2em} \geq \hspace{-0.2em}1    \nonumber \\ 
	&\bigwedge_{m\in M}\hspace{-0.4em}\Big(\hspace{-0.2em}\bigwedge_{x \in X}\hspace{-0.4em}\frac{\partial B}{\partial x}(a,x)f_m(x)  \hspace{-0.2em}+\hspace{-0.2em}\frac{1}{2}\text{Tr}\big(g_m^T(x)\frac{\partial^2 B}{\partial x^2}(a,x)g_m(x)\big) \hspace{-0.2em}\leq\hspace{-0.2em} c\Big).
	\end{align*}
    In a similar way, one can obtain a feasibility expression $\psi(a,x)$ for multiple barrier certificates.
	The coefficients $a_i$ can be efficiently found using SMT solvers such as Z3 for the finite set $\overline{X}\subset X$ of data samples. We denote the obtained candidate barrier certificate with fixed coefficients $a_i$ by $B(a,x)|_a$ and the corresponding feasibility expression by $\psi(a,x)|_a$. Next we obtain counterexample $x\in X$ such that $B(a,x)|_a$ satisfies $\neg \psi(a,x)|_a$. If $\neg \psi(a,x)|_a$ has no feasible solution, then the obtained $B(a,x)|_a$ is a true barrier certificate. If $\neg \psi(a,x)|_a$ is feasible, we update data samples as $\overline{X}=\overline{X}\cup x$ and recompute coefficients $a_i$ iteratively until $\neg \psi(a,x)|_a$ becomes infeasible. For detailed overview on CEGIS procedure to compute such barrier certificates we refer interested readers to \cite{jagtap2019formal}. To obtain a tight upper bound on the probability, one can utilize bisection method over $c$ and $\gamma$ iteratively. 
	\begin{remark}
		In addition, under the assumption that $f_m$ and $g_m$, $m\in M$ are polynomial functions of $\xi$, the conditions in theorems \ref{barrier1} and \ref{barrier2} can be formulated as a sum-of-square program to compute polynomial type barrier certificate similar to the one used in \cite{4287147}. 
	\end{remark}
	\section{Example}
	For the Example \ref{example},the obtained minimal values of $c$ and $\gamma$ for each of the elements of $\mathcal{P}^p(\textbf{q})$ and their computed upper bounds $\gamma+c T$ based on SMT solver Z3 and CEGIS approach are listed in Table \ref{tab}. Now, using Theorem \ref{upper_bound} we find that
	the lower bound on the probabilities that $\xi$ starts at any $x_0 \in L^{-1}(p)$, $p\in\Pi$ satisfying safe-LTL$_{F\backslash\medcirc}$ property \eqref{LTL_f} over time horizon $T=10$ are 
	\begin{align*}
	\mathbb{P}_{x_0}\{\sigma_{\xi} \models \varphi\} &\geq 0.99788\quad  \forall x_0 \in L^{-1}(p_0);\\ \mathbb{P}_{x_0}\{\sigma_{\xi} \models \varphi\} &\geq 0.96563\quad  \forall x_0 \in L^{-1}(p_1); \text{ and } \\
	\mathbb{P}_{x_0}\{\sigma_{\xi} \models \varphi\} &\geq 0\quad  \forall x_0 \in L^{-1}(p_1)\text{ and } \forall x_0 \in L^{-1}(p_3).
	\end{align*}
	For this computation, we used polynomial barrier certificates of order 5 each with 21 coefficients for all $\nu$. Each individual computation takes an average of 3 hours using an Intel i7-7700 processor with a 16GB RAM.
	
	\begin{table}[t]
		\caption{Values of $c$ and $\gamma$ for all $\nu\in \mathcal{P}^p(\bf{q})$, $\bf{q}\in \mathcal{R}^p$}
		\label{tab}
		\begin{tabular}{@{}llll@{}}
			\toprule
			$\nu$ & $c$            & $\gamma$ & $\gamma+cT$ \\ \midrule
			$(q_0,q_1,q_2)$    & $1.953125 \times 10^{-4}$ & $9.765 \times 10^{-5}$    & $0.002050$                 \\
			$(q_1,q_2,q_3)$    & $0.25$         & $0.25$                    & $1$                        \\
			$(q_0,q_1,q_4)$    & $1.853125 \times 10^{-4}$ & $1.853125\times 10^{-4}$            & $0.002038$                \\
			$(q_1,q_4,q_3)$    & $1.953125\times 10^{-4}$ & $9.765 \times 10^{-5}$    & $0.002050$                 \\
			$(q_0,q_4,q_3)$    & $0.003125$     & $0.003125$                & $0.003437$ \\     \bottomrule        
		\end{tabular}
	\end{table}
	\addtolength{\textheight}{-12cm}   
	


	%
	%
	%
	%
	%
	%
	%
	
	
	\bibliographystyle{IEEEtran}
	
	\bibliography{bibliography.bib}

\end{document}